\documentclass[11pt]{article}

\usepackage{amsmath,amssymb,amsthm}
\usepackage{booktabs}
\usepackage{array}
\usepackage{algorithm}
\usepackage[noend]{algpseudocode}
\usepackage{hyperref}
\usepackage{geometry}
\usepackage{url}
\usepackage{graphicx}
\newtheorem{lemma}{Lemma}
\newtheorem{remark}{Remark}

\DeclareMathOperator{\round}{round}
\DeclareMathOperator*{\argmin}{arg\,min}

\title{\bf A Faster Closest-Point Algorithm\\ for the $A_n^{*}$ Lattices}
\author{Yuriy A. Reznik\\[3pt]
        \normalsize Massachusetts Institute of Technology\\
        \normalsize \texttt{yreznik@mit.edu}}
\date{}

\begin{document}
\maketitle

\begin{abstract}
The dual root lattice $A_n^*$ is an important lattice in quantization,
coding, and estimation. It can be represented as the projection of the
integer lattice $\mathbb{Z}^{n+1}$ onto the $n$-dimensional hyperplane
whose coordinates sum to zero. This representation makes $A_n^*$ particularly
natural for quantizing simplex-constrained data, such as histograms and 
probability distributions.

This paper studies the closest-point problem for $A_n^*$: given a query
vector $\mathbf{y}$, find the lattice point $\mathbf{x}\in A_n^*$ minimizing
$\|\mathbf{y}-\mathbf{x}\|^2$. The fastest previously known method is the
linear-time algorithm of McKilliam, Clarkson, Smith, and Quinn (MCSQ), which
employs bucket sort as a core operation.

We present a faster linear-time algorithm. The key observation is that the
closest-point objective depends on the rounding residuals only through two
prefix aggregates: a count and a residual sum. Hence the elements inside
each bucket never need to be sorted, stored, or traversed. This replaces
the linked-list traversal and pointer chasing of MCSQ with a single
bucketing pass over two flat arrays with counting-sort-style accumulates. 
A scaled objective further makes most of the computation exact integer 
arithmetic, and when the input coordinates are rationals with a common 
denominator, for example, histograms or empirical distibutions, the entire 
algorithm becomes exact and integer-only.

Experiments on an Intel Core i9-13900H show speedups of about $1.8$ to
$3.0$ over MCSQ for $n=2,\dots,100$, with larger gains at higher
dimensions. The proposed algorithm is also noticeably faster than Conway 
and Sloan methods for other root lattices, including $A_n$, $D_n^*$, and $E_8$. 
An open-source implementation is available in the \texttt{fanstar} 
project~\cite{fanstar}.
\end{abstract}

\section{Introduction}

\subsection{The closest-point problem}

A lattice is a discrete set of points obtained by taking all integer
linear combinations of a set of basis vectors: if $B$ is a generator
matrix, then $L=\{B\mathbf{w}:\mathbf{w}\in\mathbb{Z}^n\}$.
Given a query vector $\mathbf{y}$, the \emph{closest-point problem}, also
called the nearest-lattice-point problem, is to find
\[
    \mathbf{x}^*
    =
    \argmin_{\mathbf{x}\in L}
    \|\mathbf{y}-\mathbf{x}\|^2 .
\]
In vector quantization, $\mathbf{x}^*$ is the minimum-distortion
reproduction point for $\mathbf{y}$. In communication over a Gaussian
channel, $\mathbf{x}^*$ is the maximum-likelihood decoded lattice point
\cite{ConwaySloane1982}.

This paper concerns the closest-point problem for the dual root lattice
$A_n^*$. We use the standard representation in which $A_n^*$ is embedded
in $\mathbb{R}^{n+1}$ but lies on the $n$-dimensional hyperplane
$H=\{\mathbf{x}\in\mathbb{R}^{n+1}:\mathbf{1}^{T}\mathbf{x}=0\}$,
where $\mathbf{1}$ is the all-ones vector. Throughout the paper we write
$N=n+1$ for the ambient dimension. Let
$Q = I-\mathbf{1}\mathbf{1}^{T}/N$
be the orthogonal projector onto $H$. In the normalization used here,
\[
    A_n^* = Q\mathbb{Z}^{N}
    =
    \left\{
        \mathbf{k}-\frac{\mathbf{1}^{T}\mathbf{k}}{N}\mathbf{1}
        :
        \mathbf{k}\in\mathbb{Z}^{N}
    \right\}.
\]
Thus every point of $A_n^*$ can be represented by an integer vector
$\mathbf{k}\in\mathbb{Z}^{N}$, followed by projection:
$\mathbf{x}=Q\mathbf{k}$. This representation is not unique, since
$Q(\mathbf{k}+c\mathbf{1})=Q\mathbf{k}$ for every integer $c$, but it is
very convenient algorithmically.

The input to the algorithms in this paper is a real vector
$\mathbf{y}\in\mathbb{R}^{N}$; the output is a lattice point
$\mathbf{x}=Q\mathbf{k}\in A_n^*$ for some integer representative
$\mathbf{k}$. Since all lattice points lie in $H$, an arbitrary input
$\mathbf{y}_0=Q\mathbf{y}_0+t\mathbf{1}$ may first be projected: for every
$\mathbf{x}\in A_n^*$,
$\|\mathbf{y}_0-\mathbf{x}\|^2=\|Q\mathbf{y}_0-\mathbf{x}\|^2+t^2N$,
and the second term is independent of $\mathbf{x}$. Therefore, unless
stated otherwise, we assume that the query vector has already been
projected and satisfies $\mathbf{1}^{T}\mathbf{y}=0$.

\subsection{Why $A_n^*$ is useful}

The lattice $A_n^*$ has several attractive geometric and computational
properties. In low dimensions it gives exceptionally efficient sphere
coverings; in fact, it gives the thinnest known covering in all dimensions
up to eight~\cite{ConwaySloaneBook}. It includes familiar low-dimensional
lattices as special cases: $A_2^*$ is equivalent to the hexagonal lattice,
and $A_3^*$ is equivalent to the body-centered cubic lattice.
The sum-zero embedding also makes $A_n^*$ natural for quantizing data with
a fixed-sum constraint, such as histograms, probability distributions, and
normalized feature vectors. Such constraints arise in source coding and
distribution quantization~\cite{Reznik2011,CVPR2010}, speech and audio
coding~\cite{Makhoul1975}, and modern machine-learning applications where
model parameters or activations are quantized under normalization
constraints~\cite{Savkin2025,OrdentlichPolyanskiy2026}. Efficient
closest-point algorithms are essential if such lattices are to be used in
practice.

\subsection{Prior closest-point algorithms for $A_n^*$}

Conway and Sloane~\cite{ConwaySloane1982} were the first to give a fast
closest-point algorithm for $A_n^*$. Their method decomposes $A_n^*$ into
cosets of the root lattice $A_n$ and quantizes to each coset. This gives an
$O(n^2\log n)$ algorithm, later improved to $O(n^2)$ using soft-decoding
techniques~\cite{ConwaySloane1986}.
Clarkson~\cite{Clarkson1999} gave an $O(n\log n)$ algorithm whose main
cost is sorting the fractional parts of the input coordinates. McKilliam,
Clarkson, and Quinn~\cite{MCQ} simplified this approach and also obtained
an $O(n\log n)$ algorithm, but with only one sort rather than two.
McKilliam, Clarkson, Smith, and Quinn~\cite{MCSQ} then observed that a
full sort is unnecessary. Their algorithm uses a one-pass bucket sort to
classify coordinates into $N=n+1$ threshold intervals, which yields an
$O(n)$ closest-point algorithm. We refer to this algorithm as MCSQ. It is
the fastest previously known method asymptotically.

\subsection{Contribution of this paper}

This paper presents a faster linear-time closest-point algorithm for
$A_n^*$. The improvement is not in asymptotic complexity: both MCSQ and the
proposed method are $O(n)$. Instead, the improvement comes from reducing
the constant factors and memory-access costs that dominate practical
runtime.

The key observation is that the closest-point objective does not require
the elements inside each bucket to be stored or traversed: it depends only
on the number of coordinates whose rounded value has been increased and on
the sum of the corresponding residuals. Therefore, instead of performing a
bucket sort with linked lists as in MCSQ, the proposed algorithm keeps only
two flat arrays per bucket---a count array and a residual-sum array---and
performs the minimization by a single prefix sweep over these arrays,
analogous to the prefix stage of counting sort, except that no sorted
output is ever materialized.

This removes the pointer-chasing memory accesses of MCSQ. The proposed
algorithm also eliminates an $O(n)$ squared-norm computation through a
reduced objective. Scaling this reduced objective by $N$ further removes
the per-candidate division by $N$ and makes its quadratic term---together
with the residual surplus $\alpha_0$ and the flip counts---exact integer
quantities for zero-sum inputs. In addition, the rounding and bucketing
passes are folded into a single pass, so that residuals are consumed as
they are produced and never stored, and the output is reconstructed by a
branchless threshold test. Finally, when the input coordinates are
rationals with a common denominator, as for histograms and empirical
types, the entire algorithm can be carried out in exact integer
arithmetic.

\subsection{Outline}

Section~\ref{sec:prelim} defines the candidate structure used by all
modern closest-point algorithms for $A_n^*$; Section~\ref{sec:mcsq}
reviews MCSQ; Section~\ref{sec:proposed} presents the proposed algorithm;
Sections~\ref{sec:complexity} and~\ref{sec:experiments} compare operation
counts and report experiments.

\section{Preliminaries}
\label{sec:prelim}

This section reviews the candidate structure behind the MCQ and MCSQ
algorithms. Let $\mathbf{y}\in H$ be the query vector. We seek
$\mathbf{x}^*=\argmin_{\mathbf{x}\in A_n^*}\|\mathbf{y}-\mathbf{x}\|^2$.
Since $A_n^*=Q\mathbb{Z}^{N}$, the problem is equivalent to finding an
integer representative $\mathbf{k}\in\mathbb{Z}^{N}$ such that
$Q\mathbf{k}$ is closest to $\mathbf{y}$.

A useful way to generate candidates is to round a shifted version of the
query vector. By $\round(x)$ we denote the nearest integer to a given real
number $x$. For a scalar shift $\lambda\in\mathbb{R}$, define
$f(\lambda)=\round(\mathbf{y}+\lambda\mathbf{1})$, with rounding applied
coordinatewise; then $Qf(\lambda)\in A_n^*$. Because
$f(\lambda+1)=f(\lambda)+\mathbf{1}$ and $Q\mathbf{1}=\mathbf{0}$, the
projected point is periodic, $Qf(\lambda+1)=Qf(\lambda)$, so it suffices
to consider $\lambda\in[0,1)$.

Let $\mathbf{k}_0=\round(\mathbf{y})$ and define the centered rounding
residuals
$\mathbf{z}_0=\mathbf{y}-\mathbf{k}_0$, $z_{0,j}\in[-\tfrac12,\tfrac12)$.
As $\lambda$ increases from $0$ to $1$, coordinate $j$ of $f(\lambda)$
increases by one exactly when $\lambda$ crosses
$\tau_j=\tfrac12-z_{0,j}$.
Thus the possible candidates are obtained from $\mathbf{k}_0$ by adding
one to coordinates in order of increasing threshold $\tau_j$. If the
thresholds are sorted, this gives a chain
$\mathbf{k}_0,\mathbf{k}_1,\dots,\mathbf{k}_n$,
where $\mathbf{k}_i$ differs from $\mathbf{k}_{i-1}$ by one coordinate
being increased by one, and one of the projected points $Q\mathbf{k}_i$ is
a closest lattice point.

For a candidate $\mathbf{k}_i$, define
$\mathbf{z}_i=\mathbf{y}-\mathbf{k}_i$. Since $\mathbf{y}\in H$, the
squared distance to the projected point can be written as
\[
    \|\mathbf{y}-Q\mathbf{k}_i\|^2
    =
    \mathbf{z}_i^{T}\mathbf{z}_i
    -
    \frac{(\mathbf{z}_i^{T}\mathbf{1})^2}{N}.
\]
For compactness define
$\alpha_i=\mathbf{z}_i^{T}\mathbf{1}$ and
$\beta_i=\mathbf{z}_i^{T}\mathbf{z}_i$.
Then the objective to minimize is
\begin{equation}
    d_i
    =
    \beta_i-\frac{\alpha_i^2}{N}.
    \label{eq:di}
\end{equation}
If candidate $i$ is obtained from candidate $i-1$ by increasing coordinate
$(i)$ by one, the residual in that coordinate decreases by one while all
others are unchanged, so
\begin{align}
    \alpha_i &= \alpha_{i-1}-1, \label{eq:alpharec}\\
    \beta_i  &= \beta_{i-1}-2z_{i-1,(i)}+1. \label{eq:betarec}
\end{align}
These recursions are the basis of the MCQ and MCSQ algorithms. The MCQ
algorithm sorts the thresholds $\tau_j$ and evaluates the candidates in
sorted order. MCSQ avoids the full sort using the following fact.

\begin{lemma}[\protect{\cite[Lemma~2]{MCSQ}}]
\label{lem:interval}
Let $Qf(\lambda_0)$ be a closest point of $A_n^*$ to $\mathbf{y}$, and let
$I$ be the maximal interval containing $\lambda_0$ on which
$f(\lambda)$ is constant. Then $|I|\ge 1/N$.
\end{lemma}

Lemma~\ref{lem:interval} implies that the optimal candidate can be found
by testing only the grid values $\lambda=i/N$, $i=0,1,\dots,N-1$.
Equivalently, the thresholds $\tau_j$ need only be assigned to one of
$N$ intervals of length $1/N$; full sorting is unnecessary.

\section{The MCSQ algorithm}
\label{sec:mcsq}

The MCSQ algorithm improves on MCQ by replacing sorting with bucketing.
Coordinate $j$ changes its rounded value when
$\lambda=\tau_j=\tfrac12-z_j$. Instead of sorting the $N$ thresholds
exactly, MCSQ assigns each threshold to one of $N$ intervals of length
$1/N$; by Lemma~\ref{lem:interval}, this coarse ordering suffices to
identify a closest lattice point.

MCSQ stores the coordinates belonging to each bucket in a linked list and
sweeps the buckets in order. When bucket $i$ is reached, all coordinates in
it are treated as having flipped; the running quantities $\alpha$ and
$\beta$ are updated using \eqref{eq:alpharec}--\eqref{eq:betarec}, and the
objective $d_i=\beta-\alpha^2/N$ is evaluated. This yields an $O(n)$
algorithm because each coordinate is inserted into one bucket and later
visited a constant number of times. However, the linked-list representation
has a practical cost: the sweep performs dependent memory accesses, since
the next address to load is stored in the current list element, which the
processor cannot easily prefetch or parallelize. This pointer chasing is
one of the main costs removed by the proposed algorithm.

Algorithm~\ref{alg:mcsq} reproduces the MCSQ method in a form convenient
for comparison. The phases that build and traverse the linked lists are
its main practical costs. In particular, the reconstruction phase walks
the winning buckets' lists again to determine which coordinates of
$\mathbf{k}$ must be incremented. The proposed algorithm avoids these list
traversals entirely.

\begin{algorithm}[t]
\caption{MCSQ: closest point in $A_n^*$ via one-pass bucket sort~\cite{MCSQ}}
\label{alg:mcsq}
\begin{algorithmic}[1]
\Require $\mathbf{y}\in\mathbb{R}^{N}$, $N=n+1$
\State $\mathbf{z}\gets\mathbf{y}-\round(\mathbf{y})$
\State $\alpha\gets\mathbf{z}^{T}\mathbf{1}$;\quad
       $\beta\gets\mathbf{z}^{T}\mathbf{z}$
\State $\texttt{bucket}[i]\gets 0$ for $i=1,\dots,N$
\For{$t\gets 1$ \textbf{to} $N$}
       \Comment{assign coordinate $t$ to a threshold bucket}
    \State $i\gets \left\lceil N\left(\tfrac12-z_t\right)\right\rceil$
    \State $\texttt{link}[t]\gets\texttt{bucket}[i]$
    \State $\texttt{bucket}[i]\gets t$
\EndFor
\State $D\gets\beta-\alpha^2/N$;\quad $m\gets 0$
\For{$i\gets 1$ \textbf{to} $N$}
       \Comment{sweep buckets in increasing threshold order}
    \State $t\gets\texttt{bucket}[i]$
    \While{$t\ne 0$}
           \Comment{linked-list traversal}
        \State $\alpha\gets\alpha-1$
        \State $\beta\gets\beta-2z_t+1$
        \State $t\gets\texttt{link}[t]$
    \EndWhile
    \State $h\gets\beta-\alpha^2/N$
    \If{$h<D$}
        \State $D\gets h$;\quad $m\gets i$
    \EndIf
\EndFor
\State $\mathbf{k}\gets\round(\mathbf{y})$
       \Comment{reconstruct winning integer representative}
\For{$i\gets 1$ \textbf{to} $m$}
    \State $t\gets\texttt{bucket}[i]$
    \While{$t\ne 0$}
        \State $k_t\gets k_t+1$
        \State $t\gets\texttt{link}[t]$
    \EndWhile
\EndFor
\State $\mu\gets(\mathbf{1}^{T}\mathbf{k})/N$
\State \Return $\mathbf{k}-\mu\mathbf{1}$
\end{algorithmic}
\end{algorithm}

\section{The proposed algorithm}
\label{sec:proposed}

\subsection{Key observation}

The proposed algorithm starts from a simple observation: the objective does
not require knowing the identities of the coordinates inside each bucket
during the sweep. It only requires two aggregate quantities. For bucket
$i$, define its population and residual sum
\[
    c_i
    =
    \#\{j:\tau_j\text{ belongs to bucket }i\},
    \qquad
    \sigma_i=\sum_{j:\tau_j\text{ belongs to bucket }i} z_j ,
\]
and the corresponding prefix sums
\[
    C_i=\sum_{\ell=1}^{i}c_\ell,
    \qquad
    Z_i=\sum_{\ell=1}^{i}\sigma_\ell .
\]
Here $C_i$ is the number of coordinates that have flipped after buckets
$1,\dots,i$ are processed, and $Z_i$ is the sum of their residuals.

\begin{lemma}[additive decomposability]
\label{lem:decomp}
Starting from $\alpha_0=\mathbf{z}_0^{T}\mathbf{1}$ and
$\beta_0=\mathbf{z}_0^{T}\mathbf{z}_0$, the closest-point objective after
processing buckets $1,\dots,i$, reduced by the constant $\beta_0$, is
\begin{equation}
    h_i
    =
    W_i-\frac{(\alpha_0-C_i)^2}{N},
    \qquad
    W_i = C_i - 2Z_i .
    \label{eq:hi}
\end{equation}
Therefore no ordering of coordinates within buckets is required.
\end{lemma}

\begin{proof}
Processing a flipped coordinate $j$ decreases $\alpha$ by one and increases
$\beta$ by $1-2z_j$. Summing these updates over all coordinates in buckets
$1,\dots,i$ gives $\alpha_i=\alpha_0-C_i$ and
$\beta_i=\beta_0+\sum(1-2z_j)=\beta_0+C_i-2Z_i$. Substituting into
$d_i=\beta_i-\alpha_i^2/N$ and dropping the constant $\beta_0$ gives the
result.
\end{proof}

Thus the algorithm needs only a count and a residual sum per bucket; it
never stores the elements of a bucket and never traverses a list. The
sweep is precisely the prefix-sum step of counting sort, except that the
prefix is consumed directly inside the minimization rather than used to
scatter a sorted array.

\subsection{A scaled, largely integer objective}
\label{sec:scaled}

Since $N>0$, minimizing $h_i$ in \eqref{eq:hi} is equivalent to minimizing
the scaled objective
\begin{equation}
    g_i
    =
    N h_i
    =
    \left(C_i - 2Z_i\right)N
    -
    \left(C_i-\alpha_0\right)^2 .
    \label{eq:gi}
\end{equation}
This form has two advantages. First, the per-candidate division by $N$
disappears. Second, for zero-sum inputs the surplus $\alpha_0$ is an exact
integer (see below), and the flip count $C_i$ is an integer by
construction, so the quadratic term $(C_i-\alpha_0)^2$ is evaluated in
exact integer arithmetic; the only real-valued quantity remaining in
\eqref{eq:gi} is $Z_i$, obtained by accumulating raw residuals $z_j$,
which is cheaper than accumulating the weights $1-2z_j$. The involved
integers are small: $|z_j|\le\tfrac12$ implies $|\alpha_0|\le N/2$, hence
$|C_i-\alpha_0|\le 3N/2$, so no overflow can occur for any practical $N$.
Note also that $g_0=-\alpha_0^2$, an exact integer.

\subsection{Proposed algorithm}
\label{sec:alg}

The proposed algorithm replaces the linked-list buckets of MCSQ with two
flat arrays: a counting-sort histogram $\texttt{bc}[i]=c_i$ and its
weighted analogue $\texttt{bz}[i]=\sigma_i$. Both are built in a single
pass folded together with the rounding step: for each coordinate, the
algorithm computes the rounded value $k_j$, the residual $z=y_j-k_j$, the
bucket index $b_j$, and the histogram updates, all at once. The residual is
consumed immediately and never stored, so no intermediate residual array is
needed.

After this pass, a prefix sweep over the buckets updates the cumulative
count and residual sum, evaluates the scaled objective \eqref{eq:gi}, and
records the best bucket index $m$. Finally, the winning candidate is
reconstructed and projected in a single branchless output pass:
\[
    x_j = k_j + [\,b_j\le m\,] - \mu,
\]
where $[\,P\,]$ denotes the Iverson bracket, equal to $1$ if $P$ holds and
$0$ otherwise. The comparison result is used directly as an arithmetic
operand, so the output pass contains no data-dependent branch.

The algorithm assumes $\mathbf{1}^{T}\mathbf{y}=0$. Under this assumption
the initial residual sum is computed exactly from the integer rounded
vector: with $\mathbf{k}=\round(\mathbf{y})$ and
$s=\mathbf{1}^{T}\mathbf{k}$,
\[
    \alpha_0
    =
    \mathbf{1}^{T}(\mathbf{y}-\mathbf{k})
    =
    -s,
\]
an exact integer. If the winning candidate flips $\phi$ coordinates, the
final integer representative has coordinate sum $s+\phi$, and the output
lattice point is $\mathbf{x}=\mathbf{k}-\frac{s+\phi}{N}\mathbf{1}$.
Algorithm~\ref{alg:proposed} gives the full method.

\begin{algorithm}[t]
\caption{Proposed closest-point algorithm for $A_n^*$.
The input satisfies $\mathbf{1}^{T}\mathbf{y}=0$.}
\label{alg:proposed}
\begin{algorithmic}[1]
\Require $\mathbf{y}\in\mathbb{R}^{N}$ with $\mathbf{1}^{T}\mathbf{y}=0$,
         $N=n+1$
\State $\texttt{bc}[i]\gets 0,\ \texttt{bz}[i]\gets 0$
       for $i=1,\dots,N$
\State $s\gets 0$
\For{$j\gets 1$ \textbf{to} $N$}
       \Comment{fused rounding and bucketing pass}
    \State $k_j\gets\round(y_j)$
    \State $z\gets y_j-k_j$
           \Comment{residual is consumed immediately, never stored}
    \State $s\gets s+k_j$
    \State $b_j\gets\left\lceil N\left(\tfrac12-z\right)\right\rceil$
    \State $\texttt{bc}[b_j]\gets\texttt{bc}[b_j]+1$
    \State $\texttt{bz}[b_j]\gets\texttt{bz}[b_j]+z$
\EndFor
\State $\alpha\gets -s$
       \Comment{exact integer residual sum}
\State $D\gets-\alpha^2$;\quad $m\gets 0$
       \Comment{$g_0$; exact integer}
\State $\mathrm{cnt}\gets 0$;\quad $\mathrm{acc}\gets 0$;\quad $\phi\gets 0$
\For{$i\gets 1$ \textbf{to} $N-1$}
       \Comment{prefix sweep; candidate $i=N$ duplicates $i=0$}
    \State $\mathrm{cnt}\gets\mathrm{cnt}+\texttt{bc}[i]$
    \State $\mathrm{acc}\gets\mathrm{acc}+\texttt{bz}[i]$
    \State $g\gets(\mathrm{cnt}-2\,\mathrm{acc})\,N-(\mathrm{cnt}-\alpha)^2$
    \If{$g<D$}
        \State $D\gets g$;\quad $m\gets i$;\quad $\phi\gets\mathrm{cnt}$
    \EndIf
\EndFor
\State $\mu\gets(s+\phi)/N$
\For{$j\gets 1$ \textbf{to} $N$}
       \Comment{branchless reconstruction fused with projection}
    \State $x_j\gets k_j+[\,b_j\le m\,]-\mu$
\EndFor
\State \Return $\mathbf{x}$
\end{algorithmic}
\end{algorithm}

\begin{remark}
No clamp is needed mathematically on line~7: residuals in
$[-\tfrac12,\tfrac12)$ give $N(\tfrac12-z)\in(0,N]$ and
$\lceil\cdot\rceil\in\{1,\dots,N\}$; a practical implementation may still
clamp to guard against floating-point rounding at the endpoints. The sweep
ends at $i=N-1$: bucket $N$ corresponds to flipping all $N$ coordinates,
i.e., to the representative $\mathbf{k}+\mathbf{1}$, and since
$Q(\mathbf{k}+\mathbf{1})=Q\mathbf{k}$ this candidate projects to the same
lattice point as $i=0$, already covered by the initialization
$D=-\alpha^2$. Consistently, coordinates in bucket $N$ are never
incremented on line~21, since $b_j\le m\le N-1$ excludes them.
\end{remark}

\subsection{Exact integer arithmetic for rational inputs}
\label{sec:rational}

The scaled objective \eqref{eq:gi} leaves only one real-valued quantity in
the entire algorithm: the residual prefix sum $Z_i$. This has a pleasant
consequence when the input is rational with a common denominator, which is
precisely the situation for histograms and empirical distributions
(``types'' in information theory): a histogram of $q$ samples over $N$
bins, shifted to the sum-zero hyperplane, has coordinates $y_j=u_j/q$ with
$u_j\in\mathbb{Z}$ and $\mathbf{1}^{T}\mathbf{u}=0$. In this case every
step of Algorithm~\ref{alg:proposed} can be carried out in exact integer
arithmetic. Rounding becomes an integer operation with residual numerators
$\zeta_j=u_j-q\,k_j\in[-q/2,\,q/2)$, so that $z_j=\zeta_j/q$; the bucket
index becomes an exact integer ceiling division,
$b_j=\lceil N(q-2\zeta_j)/(2q)\rceil$; the array $\texttt{bz}[\cdot]$
accumulates the integers $\zeta_j$, so that $\mathrm{acc}=qZ_i$ is an
integer; and the objective, scaled once more by $q$, is the exact integer
\[
    q\,g_i
    =
    \bigl(q\,C_i-2\,(qZ_i)\bigr)N
    -
    q\,(C_i-\alpha_0)^2 .
\]
All intermediate quantities are bounded by $O(qN^2)$ and fit comfortably in
64-bit integers for any practical $q$ and $N$. The algorithm therefore
becomes entirely free of floating-point operations, its result is exact,
and its behavior is bit-reproducible across platforms---a useful property
for interoperable codecs and for hardware implementations. Only the final
projection by $\mu=(s+\phi)/N$ involves a division, and the unnormalized
integer representative $(\mathbf{k},\phi)$ may be returned instead when an
exact output is desired.

\section{Complexity comparison}
\label{sec:complexity}

Both MCSQ and the proposed algorithm have linear asymptotic complexity;
the difference is in the constant factors and the memory-access pattern.
MCSQ stores each bucket as a linked list and traverses these lists during
the objective sweep and again during reconstruction. These traversals
involve dependent memory accesses: the address of the next element is not
known until the current element has been loaded. Such pointer chasing is
latency-bound and difficult for modern processors to prefetch. The
proposed algorithm stores only two flat arrays of length $N$, sweeps them
sequentially, and reconstructs using the stored bucket index $b_j$ and a
threshold test $b_j\le m$ whose result is consumed arithmetically rather
than by a branch.

The proposed algorithm also reduces and restructures the arithmetic: it
rounds each input coordinate once, whereas MCSQ rounds again during
reconstruction; it avoids computing the initial squared norm
$\beta_0=\mathbf{z}^{T}\mathbf{z}$, which is constant over all candidates;
it accumulates raw residuals $z_j$ instead of the weights $1-2z_j$; it
evaluates the scaled objective \eqref{eq:gi}, which removes the
per-candidate division by $N$ and computes the quadratic term in exact
integer arithmetic; it folds the rounding and bucketing steps into one
pass, never materializing the residual vector $\mathbf{z}$; and its output
pass is branchless, avoiding a data-dependent branch that would be
mispredicted roughly half the time.

Table~\ref{tab:counts} summarizes the main operation counts, with
floating-point and integer operations listed separately. The proposed
algorithm halves the number of rounding operations, removes two
floating-point additions and two floating-point multiplications per
coordinate, and shifts part of the remaining work to cheap, exact integer
operations. Most importantly, it eliminates all dependent pointer-chasing
lookups; the additional direct lookups are independent indexed accesses
into small flat arrays and are much more favorable on modern memory
hierarchies.

\begin{table}[t]
\centering
\caption{Worst-case operation counts, as multiples of $N=n+1$, for MCSQ
(Algorithm~\ref{alg:mcsq}) and the proposed algorithm
(Algorithm~\ref{alg:proposed}). Counts cover all phases; doubling
($2z$, $2\,\mathrm{acc}$) is counted as an addition. Both algorithms
additionally perform $N$ ceiling operations for bucket selection. Lookups
count indexed accesses to the auxiliary arrays, with a read-modify-write
counted once.}
\label{tab:counts}
\begin{tabular}{lcc}
\toprule
Operation & MCSQ & Proposed \\
\midrule
\multicolumn{3}{l}{\emph{Floating-point operations}}\\
\quad Roundings                & $2N$  & $N$   \\
\quad Additions / subtractions & $10N$ & $8N$  \\
\quad Multiplications          & $4N$  & $2N$  \\
\midrule
\multicolumn{3}{l}{\emph{Integer operations}}\\
\quad Additions / subtractions & $2N$  & $5N$  \\
\quad Multiplications          & $0$   & $N$   \\
\midrule
Comparisons                    & $N$   & $2N$  \\
Lookups: direct                & $5N$  & $6N$  \\
Lookups: \textbf{dependent}    & $6N$  & $\mathbf{0}$ \\
\bottomrule
\end{tabular}
\end{table}

\section{Experimental results}
\label{sec:experiments}

\subsection{Setup}

All results below were obtained using the implementations in the
open-source project \texttt{fanstar}~\cite{fanstar}. For each dimension
$n$, we generated $10^5$ random input vectors and mapped them into the
ambient $(n+1)$-dimensional sum-zero hyperplane by a norm-preserving
transform. We timed five $A_n^*$ closest-point implementations: 
CS~\cite{ConwaySloane1982}, VC~\cite{Clarkson1999},
MCQ~\cite{MCQ}, MCSQ~\cite{MCSQ}, and the proposed algorithm, for $n$ ranging over
$2,\dots,100$. For $n=2,\dots,8$ the proposed algorithm uses fully
unrolled implementations, switching to a generic-$n$ implementation for
$n\ge 9$. For reference and to gauge the competitiveness of $A_n^*$
against other root lattices, we also timed the Conway--Sloane
closest-point algorithms for $Z^n$, $A_n$, $D_n$, $D_n^*$, and
$E_8$~\cite{ConwaySloane1982}.

All algorithms are implemented as portable, single-threaded C89 code. The
compiler was instructed to optimize for speed but with advanced parallel
instruction sets disabled, i.e., no SSE, AVX, or similar extensions. Thus
the numbers reflect scalar performance. Timings were measured on a PC with
a 13th-generation Intel Core i9-13900H, 2.60~GHz CPU.
Table~\ref{tab:times} reports execution time in seconds per $10^5$
iterations.

\begin{table*}[t]
\centering
\caption{Execution times, in seconds per $10^5$ iterations, for different
closest-point implementations. The proposed algorithm is shown in the
``$A_n^*$: Proposed'' column.}
\label{tab:times}
\small
\begin{tabular}{r rrrrr rrrrr}
\toprule
 & \multicolumn{5}{c}{$A_n^*$ implementations}
 & \multicolumn{5}{c}{CS for other lattices} \\
\cmidrule(lr){2-6}\cmidrule(lr){7-11}
$n$ & CS & VC & MCQ & MCSQ & \textbf{Proposed}
    & $Z^n$ & $A_n$ & $D_n$ & $D_n^*$ & $E_8$ \\
\midrule
2	& 0.016	& 0.007	& 0.006	& 0.007	& \textbf{0.003} & 0.003 & 0.006 & 0.005 & 0.006         \\
3	& 0.029	& 0.006	& 0.007	& 0.007	& \textbf{0.003} & 0.003 & 0.006 & 0.005 & 0.006         \\
4	& 0.044	& 0.009	& 0.010	& 0.009	& \textbf{0.005} & 0.004 & 0.008 & 0.006 & 0.008         \\
5	& 0.060	& 0.014	& 0.010	& 0.011	& \textbf{0.005} & 0.005 & 0.009 & 0.008 & 0.009         \\
6	& 0.083	& 0.018	& 0.013	& 0.014	& \textbf{0.007} & 0.005 & 0.011 & 0.007 & 0.008         \\
7	& 0.103	& 0.021	& 0.015	& 0.016	& \textbf{0.006} & 0.006 & 0.013 & 0.009 & 0.009         \\
8	& 0.130	& 0.024	& 0.018	& 0.018	& \textbf{0.007} & 0.006 & 0.014 & 0.009 & 0.012 & 0.015 \\
9	& 0.155	& 0.029	& 0.019	& 0.020	& \textbf{0.009} & 0.008 & 0.016 & 0.010 & 0.014         \\
13	& 0.282	& 0.041	& 0.028	& 0.026	& \textbf{0.011} & 0.009 & 0.023 & 0.013 & 0.017         \\
18	& 0.453	& 0.058	& 0.038	& 0.033	& \textbf{0.013} & 0.011 & 0.031 & 0.016 & 0.023         \\
25	& 0.856	& 0.085	& 0.054	& 0.044	& \textbf{0.016} & 0.014 & 0.046 & 0.018 & 0.031         \\
35	& 1.700	& 0.132	& 0.077	& 0.059	& \textbf{0.022} & 0.020 & 0.066 & 0.024 & 0.041         \\
50	& 3.790	& 0.216	& 0.120	& 0.087	& \textbf{0.029} & 0.026 & 0.103 & 0.031 & 0.054         \\
72	& 9.110	& 0.350	& 0.191	& 0.121	& \textbf{0.042} & 0.036 & 0.158 & 0.044 & 0.076         \\
100	& 20.065& 0.545	& 0.290	& 0.165	& \textbf{0.056} & 0.048 & 0.243 & 0.058 & 0.104         \\
\bottomrule
\end{tabular}
\end{table*}

\subsection{Discussion}

In terms of runtime, the proposed method outperforms the fastest previous
$A_n^*$ algorithm, MCSQ, by approximately a factor of $1.8$ to $3.0$.
The larger speedups occur at higher dimensions, consistent with the
elimination of pointer-chasing loads whose relative cost grows as the
working set becomes less cache-local. Against the earlier CS, VC, and MCQ
methods the advantage is far larger: at $n=100$ the proposed algorithm is
roughly $5\times$ faster than MCQ, about $10\times$ faster than VC, and 
more than two orders of magnitude faster than CS.

Two comparisons against other lattice families are worth highlighting.
First, with the proposed improvements in place, $A_n^*$ becomes highly
competitive with the other root lattices commonly used as quantizers: at
large $n$ the proposed algorithm is faster than the CS closest-point
routines for $A_n$ and $D_n^*$ lattice, and about as fast as CS routines 
for $D_n*$ lattice. For $n=8$, the proposed $A_8^*$ implementation is also 
faster than the CS closest-point routine for $E_8$. The only lattice that 
is clearly faster is $Z^n$, which represents independent scalar quantizers. 

The above comparison is between the complexity of nearest-point routines, 
not between the geometric merits of the lattices themselves. Nevertheless, 
it suggests that the entire $A_n^*$ family deserves renewed attention in 
practical quantization and coding systems.

\section{Conclusions}

We presented a faster linear-time closest-point algorithm for the dual root
lattice $A_n^*$, based on the observation that the candidate objective
depends on the rounded residuals only through two prefix aggregates: a
count and a residual sum. Compared with the previously fastest MCSQ
algorithm, the proposed method replaces linked-list bucket traversal with
flat count and residual-sum arrays, eliminating pointer chasing, reducing
the rounding operations, and avoiding the initial squared-norm
computation. The reduced objective is scaled so that its quadratic term,
the residual surplus, and the flip counts are exact integer quantities for
sum-zero inputs; the rounding and bucketing steps are folded into a single
pass with no intermediate residual array; and the output is reconstructed
by a branchless threshold test fused with the final projection.

As an additional benefit, when the input coordinates are rationals with 
a common denominator, such as for histograms and empirical distributions,
the entire algorithm runs in exact integer arithmetic, producing exact, 
bit-reproducible results without a floating-point unit.

Experiments for dimensions $n=2,\dots,100$ show speedups of roughly
$1.8$ to $3.0$ over MCSQ, with larger gains at higher dimensions. With the
proposed improvements in place, finding closest points in $A_n^*$ becomes
faster than in other root lattices, such as $A_n$, $D_n^*$, and $E_8$. 
These results suggest that $A_n^*$ deserves renewed attention as a practical 
lattice for quantization, coding, and estimation problems, especially when 
the data naturally lie on or near a simplex. The implementation is available 
in the open-source project \texttt{fanstar}~\cite{fanstar}.


\begin{thebibliography}{9}

\bibitem{ConwaySloane1982}
J.~H. Conway and N.~J.~A. Sloane,
``Fast quantizing and decoding algorithms for lattice quantizers and
codes,''
\emph{IEEE Trans. Inf. Theory}, vol.~28, no.~2, pp.~227--232, Mar.~1982.

\bibitem{ConwaySloane1986}
J.~H. Conway and N.~J.~A. Sloane,
``Soft decoding techniques for codes and lattices, including the Golay
code and the Leech lattice,''
\emph{IEEE Trans. Inf. Theory}, vol.~32, no.~1, pp.~41--50, Jan.~1986.

\bibitem{ConwaySloaneBook}
J.~H. Conway and N.~J.~A. Sloane,
\emph{Sphere Packings, Lattices and Groups}, 3rd ed.
New York: Springer-Verlag, 1998.

\bibitem{Clarkson1999}
I.~V.~L. Clarkson,
``An algorithm to compute a nearest point in the lattice $A_n^{*}$,''
in \emph{Applied Algebra, Algebraic Algorithms and Error-Correcting
Codes} (Lecture Notes in Computer Science), vol.~1719.
Springer, 1999, pp.~104--120.

\bibitem{MCQ}
R.~G. McKilliam, I.~V.~L. Clarkson, and B.~G. Quinn,
``An algorithm to compute the nearest point in the lattice $A_n^{*}$,''
\emph{IEEE Trans. Inf. Theory}, vol.~54, no.~9, pp.~4378--4381,
Sep.~2008.

\bibitem{MCSQ}
R.~G. McKilliam, I.~V.~L. Clarkson, W.~D. Smith, and B.~G. Quinn,
``A linear-time nearest point algorithm for the lattice $A_n^{*}$,''
in \emph{Proc. Int. Symp. Information Theory and its Applications
(ISITA'08)}, Auckland, New Zealand, Dec.~2008, pp.~1--5.

\bibitem{Reznik2011}
Y.~Reznik,
``An algorithm for quantization of discrete probability distributions,''
in \emph{Proc. Data Compression Conference (DCC'11)},
Snowbird, UT, Mar.~2011, pp.~333--343.

\bibitem{CVPR2010}
V. Chandrasekhar, Y. Reznik, G. Takacs, D. M. Chen, S. S. Tsai, R. Grzeszczuk and B. Girod, 
``Quantization Schemes for the Compressed Histogram of Gradients descriptor,''
 in \emph{Proc. Computer Vision and Pattern Recognition (CVPR'10)}, 
San Francisco, CA, June 2010, pp.~33--40. 

\bibitem{Makhoul1975}
J. Makhoul, 
``Linear prediction: A tutorial review,'' 
in \emph{Proceedings of the IEEE}, vol. 63, no. 4, pp. 561-580, April 1975.

\bibitem{Savkin2025}
S.~Savkin,
``Quantization Methods for Matrix Multiplication and Efficient Transformers,''
M.S. thesis, Massachusetts Institute of Technology, Sep.~2025.
\url{https://people.lids.mit.edu/yp/homepage/data/theses/2025_MEng_Savkin.pdf}

\bibitem{OrdentlichPolyanskiy2026}
O.~Ordentlich and Y.~Polyanskiy,
``Optimal Quantization for Matrix Multiplication,''
\emph{IEEE Trans. Inf. Theory}, vol.~72, no.~3, pp.~1943--1972,
Mar.~2026.

\bibitem{fanstar}
Y.~Reznik,
\emph{fanstar: Faster closest-point algorithms for the $A_n^{*}$ lattices},
open-source software.
\url{https://github.com/yuriy-a-reznik/fanstar}

\end{thebibliography}
\end{document}